\documentclass[10pt, conference]{IEEEtran}
\IEEEoverridecommandlockouts 

\usepackage{tabularx}
\usepackage{booktabs}
\usepackage{graphicx}
\usepackage{subfigure}
\usepackage{color}
\usepackage{epsfig}
\usepackage{amssymb}
\usepackage{amsmath}
\usepackage{amsthm}
\usepackage{latexsym}
\usepackage{setspace}
\usepackage{bbm}
\usepackage{float}
\usepackage{dsfont}
\usepackage{caption}
\usepackage{enumerate}
\usepackage[ruled,vlined]{algorithm2e}

\newcommand{\dP}{\mathrm{P}}
\newcommand{\dQ}{\mathrm{Q}}
\newcommand{\bPP}[1]{{\dP_{#1}}}
\newcommand{\bQQ}[1]{{\dQ_{#1}}}
\newcommand{\bPr}[1]{{\mathbb{P}}\left(#1\right)}

\newcommand{\bP}[2]{\mathrm{P}_{#1}\left({#2}\right)}
\newcommand{\bQ}[2]{\mathrm{Q}_{#1}\left({#2}\right)}

\newcommand{\cA}{{\mathcal A}}

\newcommand{\cC}{{\mathcal C}}

\newcommand{\cF}{{\mathcal F}}

\newcommand{\cM}{{\mathcal M}}

\newcommand{\cO}{{\mathcal O}}
\newcommand{\cP}{{\mathcal P}}

\newcommand{\cU}{{\mathcal U}}

\newcommand{\cX}{{\mathcal X}}

\newcommand{\ep}{\varepsilon}

\newcommand{\indicator}{{\mathds{1}}}
\newcommand{\mc}{-\!\!\!\!\circ\!\!\!\!-}

\newtheorem{theorem}{Theorem}

\newtheorem*{corollary*}{Corollary}

\newtheorem{lemma}[theorem]{Lemma}
\newtheorem*{lemma*}{Lemmas}

\theoremstyle{remark}
\newtheorem*{remark*}{Remark}
\newtheorem*{remarks*}{Remarks}
\theoremstyle{definition}
\newtheorem{definition}{Definition}

\newtheorem{example}{Example}

\allowdisplaybreaks

\newcommand{\win}{\omega}
\newcommand{\numwin}{N_{\omega}}

\begin{document}

\setlength\textfloatsep{0pt}

\title{A New Proof of Nonsignalling Multiprover Parallel Repetition
  Theorem}

\author{ \IEEEauthorblockN{Himanshu Tyagi$^\dag$} \and
  \IEEEauthorblockN{Shun Watanabe$^\ddag$} }

\maketitle

{\renewcommand{\thefootnote}{}\footnotetext{
\noindent$^\dag$Department of Electrical Communication Engineering,
Indian Institute of Science, Bangalore 560012, India.  Email:
htyagi@ece.iisc.ernet.in.

\noindent$^\ddag$Department of Computer and Information Sciences,
Tokyo University of Agriculture and Technology, Tokyo 184-8588, Japan.
Email: shunwata@cc.tuat.ac.jp.  }}

\maketitle

\renewcommand{\thefootnote}{\arabic{footnote}}
\setcounter{footnote}{0}

\begin{abstract}
We present an information theoretic proof of the
nonsignalling multiprover parallel repetition theorem, a recent
extension of its two-prover variant that underlies many hardness of
approximation results. The original proofs used de Finetti type
decomposition for strategies. We present a new proof that is based on
a technique we introduced recently 
for proving strong converse results in multiuser
information theory and entails a 
change of measure after replacing hard information constraints with soft ones.
\end{abstract}

\section{Introduction}
The parallel repetition theorem is an important tool in theoretical
computer science which is used to prove hardness of approximation
results. It shows roughly that if distributed provers can
satisfy a random predicate with probability $v<1$ without coordinating,
then they can satisfy $n$ independent copies of the same predicate
only with probability going to $0$ exponentially in $n$. Such a
theorem for two-prover case was shown in~\cite{Raz98}, with a
simplified proof given in~\cite{Hol09}. The precise form of the
statement of such a theorem relies on the structure of the query distribution, the predicate,
and the class of strategies allowed for the provers. In particular, it
was noted in~\cite{Hol09} that in most applications we only need a
parallel repetition theorem for nonsignalling strategies, a class of
correlation that subsumes even quantum correlations. 

While the validity of a multiprover parallel repetition theorem for the standard
setting is unclear, recently such a theorem has been proved for the
nonsignalling setting~\cite{LancienWinter16} (see, also,
\cite{ArnRenVid16, BuhFehSch14}). The proof uses de Finetti type decomposition of
strategies and a linear programming interpretation of the value
function. In this paper, we provide a new proof of the same result
that is completely ``information theoretic". Our proof draws on the
connection between the parallel repetition setting and that of multiuser rate-distortion
theory. In particular, we rely on a change of measure approach for proving strong converse
results in multiuser distributed coding problems. This approach was introduced for centralized coding problems in \cite{GuEff09},
and was recently sophisticated and extended to distributed coding problems in~\cite{TyaWat18} 
using a relaxation technique introduced in \cite{Ooh16}; see \cite{TyaWat18} for detailed account.  
In the change of measure approach, we first
replace the hard information constraints involving conditional
independence by their soft counterparts involving bounds on KL
divergences. Next, we change measure to that obtained by conditioning
on the ``winning'' event. The $n$-fold problem is related to a single
instance of the problem using a tensorization property of the
resulting value function.

This paper is part review -- we recall the formulation and
results for two provers in Section~\ref{s:two-prover}, followed by
those for the multiprover setting in Section~\ref{s:multiprover}. Our main contribution is a new
proof of the multiprover parallel repetition theorem
(Theorem~\ref{t:multiprover_PRT}) given in Section~\ref{s:proof}. The
final section contains brief concluding remarks.

{\em Notation.} Given random variable $(X_1, ..., X_m)$, for a subset
$\cA$ of $\{1, ..., m\}$, we
abbreviate the random variable $(X_i, i\in \cA)$ as $X_\cA$. Similarly, for a tupple $(x_1, ..., x_m)$, denote $x_\cA= (x_i, i\in \cA)$.
For other notations, we basically follow \cite{CsiKor11}.

\section{Two-Prover Parallel Repetition Theorem}\label{s:two-prover}

We begin by reviewing the two-prover setting. A two-prover game $G$
consists of a verifier and two-provers $\cP_1$ and $\cP_2$.  The
verifier samples a query $(X_1,X_2)$ according to a fixed joint
distribution $\bPP{X_1X_2}$ on finite alphabet $\cX_1 \times \cX_2$,
and sends $X_1$ and $X_2$ to $\cP_1$ and $\cP_2$, respectively.  Upon
receiving the queries, $\cP_1$ and $\cP_2$ send responses $U_1 \in
\cU_1$ and $U_2 \in \cU_2$, respectively, where $U_i$ depends only on
$X_i$.  They may use any mappings $f_i$, $i=1,2$, of $X_i$ to get
$U_i$; for finite sets $\cU$ and $\cX$, denote by $\cF(\cU|\cX)$ the
set of all mappings from $f:\cX\to \cU$. The provers win the game if
$\win(X_1,X_2,U_1,U_2)=1$ for a prespecified predicate $\win:\cX_1\times\cX_2 \times \cU_1 \times \cU_2 \to \{0,1\}$.
We will represent the game $G$ by the pair $(\bPP{X_1X_2}, \win)$.
The goal of the provers is to choose mappings $(f_1, f_2)$ that maximize the winning
probability. This maximum winning probability is termed the {\em value
  of the game} and is given by
\begin{align*}
\rho(G) &:= \max\big\{ \mathbb{E}[\win(X_1,Y_1,f_1(X_1),f_2(X_2))] :
\\ &\hspace{20mm} f_1 \in \cF(\cU_1|\cX_1),f_2\in \cF(\cU_2|\cX_2)
\big\}.
\end{align*} 
In $n$ parallel repetitions of the game, the verifier samples
sequences of queries $X_1^n$ and $X_2^n$ according to the product
distribution $\mathrm{P}_{X_1X_2}^n$. The provers now respond with sequences
$U_1^n \in \cU_1^n$ and $U_2^n \in \cU_2^n$ where $U_i^n$ depends only
on $X_i^n$, $i=1,2$. They win the game if predicates for each
coordinate are satisfied, namely the predicate $\win^{\wedge n}$ for the
parallel repetition game $G^{\wedge n}$ is given by
\begin{align*}
\win^{\wedge n}(x_1^n,x_2^n,u_1^n,u_2^n) := \bigwedge_{j=1}^n
\win(x_{1j},x_{2j},u_{1j},u_{2j}),
\end{align*}
where $\bigwedge$ denotes the AND function. The value of $G^{\wedge n}$ is
defined similarly as follows:
\begin{align*}
\rho(G^{\wedge n}) &:= \max\big\{ \mathbb{E}[\win^{\wedge
    n}(X_1^n,X_2^n,f_1(X_1^n),f_2(X_2^n))] : \\ &\hspace{20mm} f_1 \in
\cF(\cU_1^n|\cX_1^n), f_2 \in \cF(\cU_2^n|\cX_2^n) \big\}.
\end{align*}
As a simple attempt towards winning the parallel repetition game,
provers may simply apply strategies for single instance of the game
across each coordinate. In fact, they may use a different strategy for
each coordinate. Clearly, any such attempt will have value less
than $\rho(G)^n$. But can they do better by using other functions
$f_i$ that take into account the entire vector $X_i^n$ and do not have
a product structure across coordinates? At a high level, a parallel
repetition theorem says that the answer is no: The exponential decay
of value with $n$ is unavoidable.

The first instance of parallel repetition theorem was shown by Raz~\cite{Raz98} (see \cite{Hol09} for simpler proof). 
\begin{theorem}[\cite{Raz98}] \label{theorem:standard}
There exists a function $C:[0,1]\to [0,1]$ satisfying $C(t)<1$ if
$t<1$ such that for any game $G$,
\begin{align*}
\rho(G^{\wedge n}) \le C(\rho(G))^{- \frac{n}{\log |\cU_1||\cU_2|}}.
\end{align*}
\end{theorem}
The statement above holds for any game $G$ with the same universal
function $C(\cdot)$ and universal exponent that depends only on the
cardinality of the response set $\cU_1\times \cU_2$.

An important aspect of the setting above, which will be a prime focus
here, is the role of randomness in response strategies. A simple
derandomization argument shows that the value of games will not change
if the pair $(f_1, f_2)$ is generated randomly using shared randomness
$V$ that is independent of the query. Such strategies with shared
randomness available to the provers can be described by channels
\begin{align} 
&\bPP{U_1U_2|X_1X_2}(u_1,u_2|x_1,x_2) \nonumber\\ &\hspace{1em}=
  \sum_{f_1 \in \cF(\cU_1|\cX_1) \atop f_2 \in \cF(\cU_2|\cX_2)}
  \mu(f_1,f_2) \delta_{f_1,f_2}(u_1,u_2|x_1,x_2)
\label{eq:HVT}
\end{align}
where $\mu$ is a distribution on
$\cF(\cU_1|\cX_1)\times\cF(U_2|\cX_2)$ and $\delta_{f_1,f_2}$ given by
$\delta_{f_1,f_2}(u_1,u_2|x_1,x_2):= \indicator_{ \{u_1 = f_1(x_1), u_2=f_2(x_2)\}}$
  is the deterministic strategy induced by functions $f_1, f_2$.

In physics, strategies of the form \eqref{eq:HVT} are said to satisfy
the {\em hidden variable theory}, a classical physics principle which
says that if all the hidden variables are revealed then the state of
the world will be deterministic.  We denote the set of all such
strategies by $\cP_{\mathtt{HVT}} = \cP_{\mathtt{HVT}}(\cU_1 \times
\cU_2|\cX_1 \times \cX_2)$. With this new notation at our disposal and using
the observation above that shared randomness does not improve the
value of a game, we can express $\rho(G)$ alternatively as
\begin{align*}
\rho(G) = \max_{\bPP{U_1U_2|X_1X_2} \in \cP_{\mathtt{HVT}}}
\mathbb{E}[\win(X_1,X_2,U_1,U_2)].
\end{align*}
Note that since strategies using shared randomness can perform at best
as deterministic strategies, the same must be true for strategies
using independent private randomness at the provers. Thus, yet another
alternative form of $\rho(G)$ is given by
\begin{align}
\rho(G) &= \max\big\{ \mathbb{E}[\win(X_1,X_2,U_1,U_2)]: \nonumber \\ &\hspace{10mm}
\bPP{U_1U_2|X_1X_2} \mbox{ s.t. } U_1 \mc X_1 \mc X_2 \mc U_2 \big\}, \label{eq:value-long-markov}
\end{align} 
namely we can consider maximization over Markov chains $U_1 \mc X_1
\mc X_2 \mc U_2$ with marginal of $(X_1,X_2)$ fixed to $\bPP{X_1X_2}$.

It is important to examine the limitation posed by restricting to
strategies in $\cP_{\mathtt{HVT}}$. In fact, a contentious debate in
physics revolving around statistical modeling 
of quantum measurements was finally settled in the second half of the
previous century through quantitative distinction between correlations
allowed in hidden variable theory and more general {\em nonsignalling}
correlation. 

For our setting, we can define the class of {\em
  nonsignaling strategies} as follows.
\begin{definition}[Nonsignalling strategies]
Let $\cP_{\mathtt{NS}} = \cP_{\mathtt{NS}}(\cU_1\times
\cU_2|\cX_1\times \cX_2)$ be the set of all strategies
$\bPP{U_1U_2|X_1X_2}$ satisfying
\begin{align} 
\begin{aligned}
\bP{U_1|X_1X_2}{u_1|x_1,x_2} 
&=  \bP{U_1|X_1 X_2}{u_1|x_1,x_2^\prime},
\\
\bP{U_2|X_1X_2}{u_2|x_1,x_2}
&=\bP{U_2|X_1X_2}{u_2|x_1^\prime,x_2}
\end{aligned}
\label{eq:NS}
\end{align}
for every $x_1 \neq x_1^\prime$ and $x_2 \neq
x_2^\prime$. Equivalently, we can express these conditions as 
$I(U_1 \wedge X_2|X_1) = I(U_2 \wedge
X_1|X_2)=0$,
namely the Markov relations $U_1\mc X_1 \mc X_2$ and $U_2\mc X_2 \mc
X_1$ hold. 
\end{definition}

Note that these strategies include as a special case the
``long Markov strategies'' satisfying $U_1\mc X_1 \mc X_2 \mc
U_2$. This latter class performs as well as $\cP_{\mathtt{HVT}}$.
In fact, it is easy to verify that strategies in $\cP_{\mathtt{HVT}}$
satisfy~\eqref{eq:NS}, which yields
\begin{align}
\cP_{\mathtt{HVT}} \subset \cP_{\mathtt{NS}}.
\label{e:HVT_in_NS}
\end{align}
In typical applications of parallel repetition
theorem in complexity theory, it suffices to use a version of the
theorem  for nonsignalling strategies. In any case, the next question
is of independent interest: Does parallel repetition theorem hold if we allow the
broader class of nonsignalling strategies?

Specifically, denote by $\rho_{\mathtt{NS}}(G)$ the maximum
probability of satisfying $\win$ using nonsignalling strategies, i.e.,
\begin{align*}
\rho_{\mathtt{NS}}(G) := \max_{\bPP{U_1U_2|X_1X_2} \in
  \cP_{\mathtt{NS}}} \mathbb{E}[\win(X_1,X_2,U_1,U_2)].
\end{align*}
By~\eqref{e:HVT_in_NS}, $\rho(G) \le \rho_{\mathtt{NS}}(G)$. In fact,
the inequality can be strict for some games as illustrated by the next
example.

\begin{example}
Consider the following CHSH type Bell test experiment \cite{CHSH69}.
For $\cX_1=\cX_2=\cU_1=\cU_2=\{0,1\}$, let $\bPP{X_1X_2}$ be the
uniform distribution on $\{0,1\}^2$, and let predicate $\win$ be given by
$\win(x_1,x_2,u_1,u_2) = \mathbf{1}[u_1 \oplus u_2 = x_1 \wedge x_2]$. It can be
seen that the winning probability of any deterministic strategy
$\delta_{f_1,f_2}$ is upper bounded by $3/4$, whereby $\rho(G) \le
3/4$. This bound is attained by the deterministic strategy
$f_1(x_1)=f_2(x_2)=0$ for all $x_1, x_2\in\{0,1\}$.

Next, consider the strategy given by
\begin{align*}
\mathrm{P}_{U_1U_2|X_1X_2}^{\mathtt{PR}}(u_1,u_2|x_1,x_2) =
\frac{1}{2} \mathbf{1}[ u_1 \oplus u_2 = x_1 \wedge x_2].
\end{align*}
This particular correlation is termed the Pepescu-Rohrlich box, PR box
for short, since it appeared in~\cite{PopRoh94}. It can be verified
that this strategy satisfies the nonsignalling
condition~\eqref{eq:NS}. But provers can win the game with probability
$1$ by using this strategy. Therefore, $\rho_{\mathtt{NS}}(G)=1$,
strictly more than $\rho(G)$.
\end{example}

Holenstein proved that the following version of parallel repetition
theorem for nonsignaling strategies.
\begin{theorem}[\cite{Hol09}] \label{theorem:non-signaling}
There exists a function $C:[0,1]\to [0,1]$ satisfying $C(t)<1$ if
$t<1$ such that for any game $G$,
\begin{align*}
\rho_{\mathtt{NS}}(G^{\wedge n}) \le C(\rho_{\mathtt{NS}}(G))^{- n}.
\end{align*}
\end{theorem}
Note that now the exponent of parallel repetition theorem doesn't even
depend on the cardinality of response set. Also, we remark that the
proof of Theorem \ref{theorem:non-signaling} in~\cite{Hol09} is much
simpler than the simplified proof of Theorem \ref{theorem:standard} in
the same paper.

\section{Multiprover Parallel Repetition Theorem}\label{s:multiprover}

Moving to the multiprover setting, a multiprover
game $G=(\bPP{X_\cM}, \win)$ consists of a verifier  
and $m$ provers $\cP_1,\ldots,\cP_m$. 
Denoting $\cM=\{1, ..., m\}$ and $X_\cM =
(X_1,\ldots,X_m)$, the verifier samples a query
$X_\cM$ according to a fixed joint distribution $\bPP{X_{\cM}}$
and sends $X_i$ to $\cP_i$ for $i$ in  $\cM$.
Upon receiving the queries, each prover $\cP_i$ sends a response $U_i
\in \cU_i$, $1\leq i \leq m$, to the verifier. 
The provers win the game if $\win(X_\cM,U_\cM) = 1$ for a given
predicate $\win:\cX_\cM \times \cU_\cM \to \{0,1\}$. 

As in the previous section, the provers' strategy can be described by 
a channel $\bPP{U_\cM|X_\cM}$.  
The set of all strategies that can be described as convex combination of deterministic, local strategies is denoted by 
$\cP_{\mathtt{HVT}} = \cP_{\mathtt{HVT}}(\cU_{\cM}|\cX_{\cM})$. 
Namely, $\cP_{\mathtt{HVT}}$ is the set of all strategies of the form
\begin{align*} 
&\bPP{U_\cM|X_\cM}(u_\cM|x_\cM) \nonumber\\ &\hspace{1em}=
  \sum_{f_i \in \cF(\cU_i|\cX_i), i\in \cM}
  \mu(f_\cM) \delta_{f_\cM}(u_\cM|x_\cM),
\end{align*}
where $\mu$ is measure on $\prod_{i=1}^m\cF(\cU_i|\cX_i)$ and $\delta_{f_\cM}(u_\cM|x_\cM)=\indicator_{\{u_i=f_i(x_i), i\in \cM\}}$. 
The value of the game that can be attained by strategies satisfying hidden variable theory is given by
\begin{align*}
\rho(G) = \max_{\bPP{U_{\cM}|X_{\cM}} \in \cP_{\mathtt{HVT}}} \mathbb{E}[ \win(X_{\cM}, U_{\cM})].
\end{align*}
The parallel repetition game $G^{\wedge n}$ is defined analogously to the two-player setting.

For the multi-prover setting, a nonsignaling strategy is a channel $\bPP{U_\cM|X_\cM}$ such that the following condition is satisfied:
\begin{align*}
\bP{U_\cA|X_\cM}{u_\cA|x_\cA, x_{\cA^c}} 
& = \bP{U_\cA|X_\cM}{u_\cA|x_\cA, {x}_{\cA^c}^\prime},
\end{align*}
for all $x_\cA, x_{\cA^c}, {x}_{\cA^c}^\prime, u_{\cA}$ and all subsets $\cA$ of $\cM$. Denoting
the set of all nonsignaling strategies by $\cP_{\mathtt{NS}} =
\cP_{\mathtt{NS}}(\cU_\cM|\cX_\cM)$, the value of the game that can be
attained by nonsignaling strategies is given by
\begin{align*}
\rho_{\mathtt{NS}}(G) = \max_{\bPP{U_{\cM}|X_{\cM}} \in
  \cP_{\mathtt{NS}}} \mathbb{E}[ \win(X_{\cM}, U_{\cM})].
\end{align*}

A general parallel repetition theorem for strategies in $\cP_{\mathtt{HVT}}$ is 
not known. As we have mentioned at the end of the previous section, proving parallel repetition theorem for strategies in $\cP_{\mathtt{NS}}$
is relatively easier than that for strategies in $\cP_{\mathtt{HVT}}$;
the former is known to hold under the condition that query distribution $\bPP{X_\cM}$ has full support \cite{ArnRenVid16, BuhFehSch14}. 
Remarkably, without the full support condition, 
a counterexample appeared in~\cite{HolYan18} for a parallel repetion theorem for $\cP_{\mathtt{NS}}$. 
Specifically, it was shown that the following three-prover game
satisfies $\rho_{\mathtt{NS}}(G^{\wedge n})=2/3$ for all $n$ (this
example appeared first in~\cite{ArnRenVid16}):
\begin{example}[Anticorrelation game]\label{e:counterexample}
For $m=3$, let $\cX_i =\cU_i=\{0,1\}$, $1\leq i \leq 3$. Let query
distribution $\bPP{X_\cM}$ be uniform on all strings $(x_1, x_2, x_3)$
with Hamming weight $2$. The required game $G$ is given by predicate
\[
\win(x_\cM, u_\cM)= \begin{cases} 1, \quad &u_i=u_j \text{ if } x_i=x_j
  \\ 0, \quad &\text{otherwise},
\end{cases}
\]
i.e., the responses are identical at the two locations where queries
are $1$.
\end{example}
This example rules out a parallel repetition theorem for
$\cP_{\mathtt{NS}}$ in general.  In other words,
$\rho_{\mathtt{NS}}(G) <1$ is not sufficient to claim the exponential
decay of winning probability in parallel repetition games.  In fact,
even preceeding this counterexample, a parallel repetition theorem, i.e.,
exponential decay, was shown to hold if the value of the single game
for a broader class of strategies, called {\em sub-nonsignalling}
strategies, is strictly less than $1$~\cite{LancienWinter16}.

Sub-nonsignaling strategies $\bPP{U_\cM|X_\cM}$, which we define next, need not be
  conditional distributions and are only required to be subnormalized,
  namely we only need it to be nonnegative and satisfying $\sum_{u_\cM}\bPP{U_\cM|X_\cM}(u_\cM|x_\cM)\leq
  1$. Both total variation distances and KL
  diveregence can be applied to such subnormalized distribution. We
  remark that the marginal $\bPP{Y}$ and the conditional distribution
  $\bPP{Y|X}$, respectively, for a subnormalized distribution
  $\bPP{XY}$ are defined as $\bP{Y}{y}=\sum_{x}\bP{XY}{x, y}$ and
  $\bP{Y|X}{y|x}=\bP{XY}{x,y}/\bP{Y}{y}$. While $\bPP Y$, too, is a
  subnormalized distribution, $\bPP{Y|X}$ will be a (normalized) distribution.

\begin{definition}[Sub-nonsignalling strategies]
The set $\cP_{\mathtt{SNS}}$ of sub-nonsignalling strategies consists
of subnormalized $\bPP{U_\cM|X_\cM}$ such that, for each subsets $\cA$ of $\cM$,
there exists a channel $\bQQ{U_\cA|X_\cA}$ satisfying:
\begin{align} \label{eq:sub-nonsignaling-condition}
{ \bP{U_\cA|X_\cM}{u_\cA|x_\cA, x_{\cA^c}} } \leq
\bQ{U_\cA|X_\cA}{u_\cA|x_\cA},
\end{align}
for all $x_\cA, x_{\cA^c}, u_{\cA}$.
\end{definition}
Note that nonsignalling strategies are those for which the inequality
condition above is replaced with identity. Heuristically,
sub-nonsignalling strategies maybe regarded as the class of strategies
close to nonsignalling stratetigies in statistical distance.  Another heuristic was suggested in~\cite{LancienWinter16} which interpretted sub-nonsignalling strategies
as nonsignalling strategies with addition $x_\cM$ dependent power to
randomly abstain from responding. In fact, we can find a
sub-nonsignalling strategy close to a distribution for which all
conditional distributions $\bPP{U_\cA|X_\cM}$ are close to some conditional distributions 
$\bQQ{U_\cA|X_\cA}$.\footnote{Lemma~\ref{lemma:approximate-sub} is a multiprover
  extension of \cite[Lemma 9.5]{Hol09} which showed that in the
  two-prover setting we can find a nonsignalling
  $\mathrm{P}^\prime_{U_\cM|X_\cM}$.}  
\begin{lemma}[{\cite[Lemma 5.2]{LancienWinter16}}]\label{lemma:approximate-sub}
Let $\bPP{X_\cM}$ be a query distribution on $\cX_\cM$, and let
$\bPP{\tilde{U}_\cM \tilde{X}_\cM}$ be a probability distribution on
$\cU_\cM \times \cX_\cM$. Suppose that for each $\cA \subsetneq \cM$ there exist a conditional distribution $\bQQ{U_\cA|X_\cA}$ such that 
\begin{align*}
d_{\mathtt{var}}(\bPP{\tilde{U}_\cA \tilde{X}_\cM}, \bPP{X_\cM} \bQQ{U_\cA|X_\cA}) \le \ep_\cA.
\end{align*}
Then, there exists a sub-nonsignaling
$\mathrm{P}^\prime_{U_\cM|X_\cM}$ such that  
\begin{align*}
d_{\mathtt{var}}(\bPP{\tilde{U}_\cM\tilde{X}_\cM}, \bPP{X_\cM} \mathrm{P}^\prime_{U_\cM|X_\cM}) \le \ep_\emptyset + 2 \sum_{\emptyset \neq \cA \subsetneq \cM} \ep_\cA.
\end{align*}
\end{lemma}

By definition, the value of the game that is attained by
sub-nonsignaling strategies satisfy 
$\rho_{\mathtt{SNS}}(G) \ge \rho_{\mathtt{NS}}(G)$. For two-prover games,
$\rho_{\mathtt{SNS}}(G)$ was shown in~\cite{LancienWinter16} to
coincide with $\rho_{\mathtt{NS}}(G)$. 
However, equality may not hold for multiprover games, in
general. Indeed, the game in Example~\ref{e:counterexample} has
$\rho_{\mathtt{NS}}(G) = 2/3$ and $\rho_{\mathtt{SNS}}(G)= 1$.
 Interestingly, when the query distribution $\bPP{X_\cM}$ has full
 support, there exists a constant $\Gamma=\Gamma(\bPP{X_\cM})$ such
 that, for $\varepsilon > 0$, ($cf$.~\cite{LancienWinter16})
\begin{align} \label{eq:SNS-full-support}
\rho_{\mathtt{NS}}(G) < 1 - \varepsilon \Longrightarrow
\rho_{\mathtt{SNS}}(G) < 1 - \frac{\varepsilon}{\Gamma}. 
\end{align} 

Before we state the parallel repetition theorem for sub-nonsignalling
strategies, we switch to a slightly more general formulation where in
the $n$ parallel repition game, instead of  
 winning all the games, we are interested in quantifying the
 probability that the provers win more than a fraction $\Delta$ of the
 game.  
This formulation is closer to the rate-distortion theory formulation 
of information theory and appeared, for instance, in 
\cite{Rao11}. Specifically, for $0 < \Delta \le 1$, consider  
\begin{align}
\rho_{\mathtt{SNS}}(G^n,\Delta) &:= \max\big\{
\Pr\big( \numwin(X_\cM^n,U_\cM^n) \ge n \Delta \big) : \nonumber \\
&\hspace{20mm} \bPP{U_\cM^n|X_\cM^n} \in \cP_{\mathtt{SNS}}(\cU_\cM^n|\cX_\cM^n)
\big\},
\label{eq:additive-predicate}
\end{align}
where $\numwin(x_\cM^n,u_\cM^n) := \sum_{j=1}^n \win(x_{\cM,j},u_{\cM,j})$.
Since $\win(x_{\cM,j}, u_{\cM,j})$ is the indicator for a win in the
$i$th coordinate, $\numwin(x_\cM^n,u_\cM^n)$ denotes the total number
of wins. Analogously, ${\cP_{\mathtt{NS}}}$ is defined by restricting
the maximum in~\eqref{eq:additive-predicate} to nonsignalling
strategies; 
our original definition $\rho_{\mathtt{NS}}(G^{\wedge n})$ coincides
with $\rho_{\mathtt{NS}}(G^n,1)$.

We now recall the multiprover parallel repetition theorem
from~\cite{LancienWinter16}.  

\begin{theorem}\label{t:multiprover_PRT}
Let $G= (\bPP{X_\cM}, \win)$ be a multiprover game with
$\rho_{\mathtt{SNS}}(G) < 1$. For any $\Delta \ge
\rho_{\mathtt{SNS}}(G)+\nu$ 
with $0 < \nu \le 1 - \rho_{\mathtt{SNS}}(G)$, we have
\begin{align*}
\rho_{\mathtt{SNS}}(G^n,\Delta) \le \exp\bigg( - n \frac{\nu^2}{C_m} \bigg),
\end{align*}
where the constant $C_m = \cO(2^{2m})$ depends only on $m=|\cM|$.
\end{theorem}
For multiprover games with full support query distributions, Theorem
\ref{t:multiprover_PRT} together with~\eqref{eq:SNS-full-support} implies the parallel
repetition theorem for nonsignaling strategies, shown first
in~\cite{BuhFehSch14}.  

The proof of the multiprover parallel repetition 
theorem for nonsignalign strategies and 
full support query distribution in \cite{BuhFehSch14}
entails extending the proof approach 
for the two-prover setting in \cite{Hol09}. An alternative proof was 
provided in \cite{ArnRenVid16} 
by using a technique based on de Finetti theorem. At a high level, this
technique allows us to restrict attention to convex combinations
of product strategies.  
In \cite{LancienWinter16}, the parallel repetition theorem for
sub-nonsignaling strategies, namely Theorem~\ref{t:multiprover_PRT}, 
was proved  by using another variant of de Finetti theorem. 

In the next section, we provide an alternative proof of Theorem
\ref{t:multiprover_PRT}.  
Our proof is based on a technique recently developed by the authors in
\cite{TyaWat18} to prove strong converse theorems for 
multi-user information theory problems. A crucial observation is that
the parallel repetition theorem can be regarded as  
an exponential strong converse of a multi-user rate-distortion problem
with no communication. In contrast to the proof
in~\cite{LancienWinter16} that uses a structural decomposition of
strategies, our proof is completely ``information theoretic". 

\section{A New Proof of Theorem~\ref{t:multiprover_PRT}}\label{s:proof}
Our goal is to derive an upper bound for the maximum probability of the
event $\cC$ of winning more than $\Delta$ fraction of games.  
 Following~\cite{TyaWat18}, we start with a change of measure (query
 distribution and provers' 
strategy) by conditioning on $\cC$. The ``distance'' between
the new distribution and the original distribution are bounded in terms of the
exponent of the probability of $\cC$. However, since we have conditioned the
strategy on the winning event, the information structure may break
down -- we are only guaranteed to be ``close'' to a distribution
satisfying our original information constraints. Nonetheless, to
complete the proof we need an appropriate single-letterization argument to
relate this new game to one instance of the original game. 

To enable this, our proof looks at the expected number of
wins instead of the probability of winning. For a given multiprover game $G = (\bPP{X_\cM}, \win)$ and $\delta \ge
0$, define
\begin{align*}
  \eta_{\mathtt{NS}}(G,\delta) &:= \max\big\{
\mathbb{E}[ \win(\tilde{X}_\cM, \tilde{U}_\cM)] : \\
&\hspace{-3mm} I(\tilde{U}_\cA \wedge \tilde{X}_{\cA^c}|\tilde{X}_\cA) + D(\bPP{\tilde{X}_\cM} \| \bPP{X_\cM}) \le \delta, \forall \cA \subsetneq \cM.
\big\}
\end{align*}
Note that the maximum is over the set of distributions, which
we call 
{\em $\delta$-approximate
  nonsignaling distributions}, that satisfy the information structure only
approximately. In particular, we have replaced the hard information
constraints required by nonsignalling strategies by their soft
counterparts expressed by bounds on KL diveregence. Below we shall see
two properties of 
$\eta_{\mathtt{NS}}(G,\delta)$: it tensorizes
and
 can be bounded above roughly by  $\rho_{\mathtt{SNS}}(G)$. We note
 that a linear programming based notion of approximate nonsignaling
 strategies was used in~\cite{Hol09, BuhFehSch14, ArnRenVid16,
   LancienWinter16}. Our divergence based notion of approximation
 is amenable to tensorization and facilitates an information theoretic proof.

We can now apply our proof recipe outlined earlier. 
Under the changed measure obtained by conditioning on $\cC$, the
expected number of wins is more than $n\Delta$. Also, this new measure
satisfies the soft information constraint bound with $\delta$ equal to
the exponent of probability of $\cC$. Thus,
$\eta_{\mathtt{NS}}(G^n,\delta)$ must be more than $n\Delta$. Using
the properties of $\eta_{\mathtt{NS}}(G^n,\delta)$ mentioned earlier, we
can bound it above roughly by  $n\rho_{\mathtt{SNS}}(G)$, which shows
that $\Delta$ must be roughly bounded above by
$\rho_{\mathtt{SNS}}(G)$. The required bound for exponent is obtained
by the contrapositive statement.

 Formal arguments follow. We begin with the tensorization property.
\begin{lemma} \label{lemma:tensorization}
For a given multiprover game $G = (\bPP{X_\cM}, \win)$, $n \in \mathbb{N}$ and $\delta \ge 0$, we have
\begin{align*}
\eta_{\mathtt{NS}}(G^n, n\delta) = n \cdot \eta_{\mathtt{NS}}(G,\delta).
\end{align*}
\end{lemma}
\begin{proof}
The inequality $\eta_{\mathtt{NS}}(G^n, n\delta) \ge n
\eta_{\mathtt{NS}}(G,\delta)$ holds by definition. For the other
direction, fix a
$n\delta$-approximate
nonsignalling distribution $\bPP{\tilde{U}_\cM^n
  \tilde{X}_\cM^n}$. We have
\begin{align}
\mathbb{E}[ \numwin(\tilde{X}_\cM, \tilde{U}_\cM)] &= \sum_{j=1}^n \mathbb{E}[ \win(\tilde{X}_{\cM,j}, \tilde{U}_{\cM,j})] \nonumber \\
&= n \mathbb{E}[ \win(\tilde{X}_{\cM,J}, \tilde{U}_{\cM,J})],  \label{eq:additivity}
\end{align}
where $J$ is distributed uniformly on $\{1,\ldots,n\}$. Furthermore,
\begin{align}
n \delta &\ge I(\tilde{U}_\cA^n \wedge \tilde{X}_{\cA^c}^n |
\tilde{X}_\cA^n) + D(\bPP{\tilde{X}_\cM^n} \| \bPP{X_\cM^n}) \nonumber 
\\
&\ge n\big[ H(\tilde{X}_{\cA^c,J} | \tilde{X}_{\cA,J}) + D(\bPP{\tilde{X}_{\cM,J}} \| \bPP{X_\cM}) \big] \nonumber \\
&~~~ - \sum_{j=1}^n H(\tilde{X}_{\cA^c,j} | \tilde{X}_\cA^n,\tilde{U}_\cA^n) \nonumber \\
&\ge n\big[ H(\tilde{X}_{\cA^c,J} | \tilde{X}_{\cA,J}) + D(\bPP{\tilde{X}_{\cM,J}} \| \bPP{X_\cM}) \big] \nonumber \\
&~~~ - \sum_{j=1}^n H(\tilde{X}_{\cA^c,j} | \tilde{X}_{\cA,j},\tilde{U}_{\cA,j}) \nonumber \\
&= n\big[ H(\tilde{X}_{\cA^c,J} | \tilde{X}_{\cA,J}) + D(\bPP{\tilde{X}_{\cM,J}} \| \bPP{X_\cM}) \big] \nonumber \\
&~~~ - n H(\tilde{X}_{\cA^c,J} | \tilde{X}_{\cA,J},\tilde{U}_{\cA,J}, J) \nonumber \\
&\ge n \big[ I(\tilde{U}_{\cA,J} \wedge \tilde{X}_{\cA^c,J}|
  \tilde{X}_{\cA,J}) + D(\bPP{\tilde{X}_{\cM,J}} \| \bPP{X_\cM})
  \big], 
\nonumber
\end{align}
where the first inequality follows from \cite[Proposition 1]{TyaWat18}
and the second and the third inequalities hold since conditioning decreases
entropy. Thus,
$\bPP{\tilde{U}_{\cM,J}, \tilde{X}_{\cM,J}}$ is a
$\delta$-approximate nonsignalling distribution and the claim follows by \eqref{eq:additivity}.
\end{proof}

Next, we relate $\eta_{\mathtt{NS}}(G,\delta)$ and
$\rho_{\mathtt{SNS}}(G)$ using Lemma~\ref{lemma:approximate-sub}.
\begin{lemma} \label{lemma:approximate-sub-2}
For a given multiprover game $G = (\bPP{X_\cM}, \win)$ and $\delta \ge 0$, we have
\begin{align*}
\eta_{\mathtt{NS}}(G,\delta) \le \rho_{\mathtt{SNS}}(G) + C_m^\prime \sqrt{(2\ln 2) \delta},
\end{align*}
where the constant $C_m^\prime = \cO(2^m)$ depends only on $m=|\cM|$.
\end{lemma}
\begin{proof}
Consider a $\delta$-approximate
nonsignalling distribution $\bPP{\tilde{U}_\cM
  \tilde{X}_\cM}$. For any $\cA \subsetneq \cM$, since $I(\tilde{U}_\cA \wedge
\tilde{X}_{\cA^c} | \tilde{X}_\cA) = D(\bPP{\tilde{U}_\cA
  \tilde{X}_\cM} \| \bPP{\tilde{X}_\cM}
\bPP{\tilde{U}_\cA|\tilde{X}_\cA} )\le \delta$ 
and $D(\bPP{\tilde{X}_\cM} \| \bPP{X_\cM}) \le \delta$, by using Pinsker's
inequality \cite{CsiKor11} and the triangle inequality, 
 we get
\begin{align*}
d_{\mathtt{var}}(\bPP{\tilde{U}_\cA \tilde{X}_\cM}, \bPP{X_\cM} \bPP{\tilde{U}_\cA|\tilde{X}_\cA}) \le \sqrt{(2\ln 2) \delta}.
\end{align*}
Next, by applying Lemma \ref{lemma:approximate-sub} with $\ep_\cA
=  \sqrt{(2\ln 2) \delta}$, there exists a sub-nonsignaling strategy $\mathrm{P}^\prime_{U_\cM|X_\cM}$ such that  
\begin{align*}
d_{\mathtt{var}}(\bPP{\tilde{U}_\cM \tilde{X}_\cM}, \bPP{X_\cM} \mathrm{P}^\prime_{U_\cM|X_\cM}) \le (2^{|\cM|+1}-3)  \sqrt{(2\ln 2) \delta}.
\end{align*}
Finally, since $\win$ is bounded by $1$, we have 
\begin{align*}
&\lefteqn{ \mathbb{E}_{\bPP{\tilde{X}_\cM \tilde{U}_\cM}} [ \win(X_\cM, U_\cM)] }
\\
&\leq
\mathbb{E}_{\bPP{X_\cM}\mathrm{P}^\prime_{\tilde{U}_\cM|\tilde{X}_\cM}}[ \win(X_\cM, U_\cM)]
\\&\hspace{2cm}
+2 d_{\mathtt{var}}(\bPP{\tilde{U}_\cM \tilde{X}_\cM}, \bPP{X_\cM} \mathrm{P}^\prime_{U_\cM|X_\cM}) \\
&\le \rho_{\mathtt{SNS}}(G) + 2 (2^{|\cM|+1}-3)  \sqrt{(2\ln 2) \delta},
\end{align*}
where the final inequality uses the fact that
$\mathrm{P}^\prime_{\tilde{U}_\cM|\tilde{X}_\cM}$ is sub-nonsignalling. We obtain the claimed
bound with $C_m^\prime = 2 (2^{m+1}-3)$ since $\bPP{\tilde{U}_\cM
  \tilde{X}_\cM}$ was an arbitrary $\delta$-approximate nonsignalling distribution.
\end{proof}

We have all the tools for the proof of Theorem~\ref{t:multiprover_PRT} in place.
\paragraph*{Proof of Theorem \ref{t:multiprover_PRT}} 
If $\rho_{\mathtt{SNS}}(G^n,\Delta) > \exp(- n \delta)$, we can 
find a sub-nonsignalling strategy $\bPP{U_\cM^n|X_\cM^n}$ such
that $\bPr{ N_\omega(U_\cM, X_\cM^n) \geq n\Delta}> \exp(-n\delta)$ for some $\delta > 0$.
Denoting
\begin{align*}
\cC = \big\{ (u_\cM^n, x_\cM^n) : \numwin(x_\cM^n, u_\cM^n) \ge n \Delta \big\},
\end{align*}
we change the measure by conditioning on the event $(U_\cM^n, X_\cM^n)
\in \cC$ as follows:\footnote{Although $\bPP{U_\cM X_\cM}$ is only a
  subnormalized distribution, the changed measure $\bPP{\tilde{U}_\cM
    \tilde{X}_\cM}$ is a distribution.}
\begin{align*}
\mathrm{P}_{\tilde{U}_\cM^n \tilde{X}_\cM^n}(u_\cM^n, x_\cM^n)
 = \frac{\mathrm{P}_{U_\cM^n X_\cM^n}(u_\cM^n,x_\cM^n) \mathbf{1}[(u_\cM^n,x_\cM^n) \in \cC)]}{\mathrm{P}_{U_\cM^n X_\cM^n}(\cC)}.
\end{align*}
Then, by a simple calculation, we have
\begin{align*} 
D(\bPP{\tilde{U}_\cM^n \tilde{X}_\cM^n} \| \bPP{U_\cM^n X_\cM^n}) = \log \frac{1}{\mathrm{P}_{U_\cM^n X_\cM^n}(\cC)} 
\le n \delta.
\end{align*}
Furthermore, for each $\cA \subsetneq \cM$, 
denoting by $\bQQ{U_\cA^n|X_\cA^n}$ the dominating conditional
distribution for the sub-nonsignaling strategy $\bPP{U_\cA^n|X_\cA^n}$ (cf.~\eqref{eq:sub-nonsignaling-condition}),
we have
\begin{align*}
\lefteqn{
I(\tilde{U}_\cA^n \wedge \tilde{X}_{\cA^c}^n | \tilde{X}_\cA^n) + D(\bPP{\tilde{X}_\cM^n} \| \bPP{X_\cM^n}) 
} \\
&\le I(\tilde{U}_\cA^n \wedge \tilde{X}_{\cA^c}^n | \tilde{X}_\cA^n) + D(\bPP{\tilde{U}_\cA^n|\tilde{X}_\cA^n} \| \bQQ{U_\cA^n|X_\cA^n} | \bPP{\tilde{X}_\cA^n}) \\
&~~~ + D(\bPP{\tilde{X}_\cM^n} \| \bPP{X_\cM^n}) \\
&= D(\bPP{\tilde{U}_\cA^n|\tilde{X}_\cM^n} \| \bPP{\tilde{U}_\cA^n|\tilde{X}_\cA^n} | \bPP{\tilde{X}_\cM^n})
 + D(\bPP{\tilde{U}_\cA^n|\tilde{X}_\cA^n} \| \bQQ{U_\cA^n|X_\cA^n} | \bPP{\tilde{X}_\cM^n}) \\
 &~~~ + D(\bPP{\tilde{X}_\cM^n} \| \bPP{X_\cM^n}) \\
&= D(\bPP{\tilde{U}_\cA^n|\tilde{X}_\cM^n} \| \bQQ{U_\cA^n|X_\cA^n} | \bPP{\tilde{X}_\cM^n}) + D(\bPP{\tilde{X}_\cM^n} \| \bPP{X_\cM^n}) \\
&\le D(\bPP{\tilde{U}_\cA^n|\tilde{X}_\cM^n} \| \bPP{U_\cA^n|X_\cM^n} | \bPP{\tilde{X}_\cM^n}) + D(\bPP{\tilde{X}_\cM^n} \| \bPP{X_\cM^n}) \\
&= D( \bPP{\tilde{U}_\cA^n \tilde{X}_\cM^n} \| \bPP{U_\cA^n X_\cM^n}) \\
&\le D( \bPP{\tilde{U}_\cA^n \tilde{X}_\cM^n} \| \bPP{U_\cA^n X_\cM^n}) \\
&~~~ + D(\bPP{\tilde{U}_{\cA^c}^n|\tilde{U}_\cA^n \tilde{X}_\cM^n} \| \bPP{U_{\cA^c}^n|U_\cA^n X_\cM^n} | \bPP{\tilde{U}_\cA^n \tilde{X}_\cM^n}) \\
&= D(\bPP{\tilde{U}_\cM^n \tilde{X}_\cM^n} \| \bPP{U_\cM^n X_\cM^n}) \\
&\le n \delta,
\end{align*}
where the second inequality follows from the sub-nonsignaling
condition \eqref{eq:sub-nonsignaling-condition} and the third inequality uses
the fact that $\bPP{U_{\cA^c}^n|U_\cA^n X_\cM^n}$ is a conditional distribution. The above bound implies that
the changed measure $\mathrm{P}_{\tilde{U}_\cM^n \tilde{X}_\cM^n}$ is $\delta$-approximate nonsignalig distribution.
Furthermore, since $\numwin(\tilde{X}_\cM,\tilde{U}_\cM^n) \ge n \Delta$ holds with probability $1$ under the changed measure $\bPP{\tilde{U}_\cM^n \tilde{X}_\cM^n}$, we have
\begin{align*}
n \Delta \le \mathbb{E}[ \numwin(\tilde{U}_\cM^n, \tilde{X}_\cM^n)] 
\le \eta_\mathtt{NS}(G^n, n \delta),
\end{align*}
which together with Lemma \ref{lemma:tensorization} and Lemma \ref{lemma:approximate-sub-2} implies 
\begin{align*}
\Delta \le \eta_{\mathtt{NS}}(G, \delta) 
\le \rho_{\mathtt{SNS}}(G) + C^\prime_m \sqrt{(2 \ln 2) \delta}.
\end{align*}
By considering contraposition, if 
\begin{align} \label{eq:violation-Delta}
\Delta > \rho_{\mathtt{SNS}}(G) + C^\prime_m \sqrt{(2 \ln 2) \delta},
\end{align}
then we have $\rho_{\mathtt{SNS}}(G^n,\Delta) \le \exp(- n \delta)$.
Thus, by setting $\delta = \frac{\nu^2}{(2\ln 2) (C_m^\prime+1)^2}$, $\Delta \ge \rho_{\mathtt{SNS}}(G) + \nu$ implies \eqref{eq:violation-Delta},
and we have the claim of the theorem. \qed

\section{Discussion}


A multiprover parallel repetition theorem for standard strategies,
i.e., strategies satisfying the hidden 
variable theory, is not available. In fact, our initial attempt 
in this work was to provide an alternative proof of the two-prover
parallel repetition theorem for the standard strategies. 
We tried to prove a counterpart of the tensorization property,
Lemma \ref{lemma:tensorization}, for standard strategies. However, 
our preliminary attempt failed, mainly because it was difficult to
identify a suitable soft constraint for the long Markov chain 
in~\eqref{eq:value-long-markov}. Nonetheless, we do believe that 
our measure change approach can be used to obtain a parallel repetition
theorem for standard strategies, 
perhaps by proving an approximate tensorization property of the value
function with suitable soft constraints.

\bibliography{IEEEabrv,references} \bibliographystyle{IEEEtranS}

\end{document}